\documentclass[11pt,a4paper]{article}

\usepackage{graphicx}
\graphicspath{{figs/}}

\usepackage[english]{babel}
\usepackage{latexsym}
\usepackage{url}
\usepackage{amssymb}
\usepackage{amsmath}
\usepackage{amsthm}
\usepackage{amsfonts}
\usepackage{ifthen,stmaryrd}
\usepackage{xspace}

\usepackage[left=1.2in,top=1.2in,right=1.2in,bottom=1.2in,nohead]{geometry}

\newtheorem{theorem}{Theorem}
\newtheorem{corollary}[theorem]{Corollary}
\newtheorem{lemma}[theorem]{Lemma}

\newcommand{\RR}{\ensuremath{\mathbb R}}  

\newcommand\T{\mathbb{T}}

\def\dart#1#2{#1\mathord\shortrightarrow#2}

\def\w{\widehat w_p}

\def\StackSPT{{\sc StackSPT}\xspace}

\def\DEF#1{\textbf{\emph{#1}}}

\begin{document}

\title{Stackelberg Shortest Path Tree Game, Revisited\thanks{This work has been partially financed 
		by the Slovenian Reseedgeh Agency, program P1-0297, project J1-4106, and
		within the EUROCORES Programme EUROGIGA (project GReGAS) of the European Science Foundation.}}

\author{Sergio Cabello\thanks{Department of Mathematics, IMFM, and 
				Department of Mathematics, FMF, University of Ljubljana, Slovenia.	
				email: {\tt sergio.cabello@fmf.uni-lj.si}}
}

\date{\today}

\maketitle

\begin{abstract}
Let $G(V,E)$ be a directed graph with $n$ vertices and $m$ edges.
The edges $E$ of $G$ are divided into two types: $E_F$ and $E_P$.
Each edge of $E_F$ has a fixed price. The edges of $E_P$ are the priceable
edges and their price is not fixed a priori. Let $r$ be a vertex of $G$.
For an assignment of prices to the edges of $E_P$, the revenue is given
by the following procedure: select a shortest path tree $T$ from $r$ with respect 
to the prices (a tree of cheapest paths); the revenue is the sum, over all priceable edges $e$,
of the product of the price of $e$ and the number of vertices below $e$ in $T$.

Assuming that $k=|E_P|\ge 2$ is a constant,
we provide a data structure whose construction takes $O(m+n\log^{k-1} n)$ time and
with the property that, when we assign prices to the edges of $E_P$,
the revenue can be computed in $(\log^{k-1} n)$.
Using our data structure, we save almost a linear factor when computing
the optimal strategy in the Stackelberg shortest paths tree game of 
[D. Bil{\`o} and L. Gual{\`a} and G. Proietti and P. Widmayer. 
Computational aspects of a 2-Player Stackelberg shortest paths tree game. Proc. WINE 2008]. 
\end{abstract}

\section{Introduction}

A \emph{Stackelberg game} is an extensive game with two players and perfect information in which 
the first player, the \emph{leader}, chooses her action and then the second player, the \emph{follower},
informed of the leader's choice, chooses her action; see~\cite[Section 6.2]{or-94}.
In a \emph{Stackelberg pricing game in networks}, the leader owns a subset of the edges in a network
and has to choose the price of those edges to maximize its revenue. The other edges of the network have
a price already fixed. The follower chooses a subnetwork of minimum price with a prescribed property, 
like for example being a spanning tree or spanning two vertices.
The revenue of the leader is determined by the prices of the edges that the follower uses in its chosen subnetwork, 
possibly combined with the amount of use of each edge.

Stackelberg network pricing games were first studied by Labb{\'e} et al~\cite{lms-98} 
when the follower is interested in a cheapest path connecting two given vertices. 
They showed that even such ``simple" problem is NP-hard when the number of priceable edges is not bounded.
There has been much follow up research; we refer the reader to the overview by van Hoesel~\cite{h-08}. 
The case when the follower is interested in a cheapest spanning tree was introduced by Cardinal et al.~\cite{cetal-11}.
Bil{\`o} et al.~\cite{bgpw-08} considered the case when the follower is interested in a shortest path tree from a prespecified root $r$
and the revenue of a priceable edge is the product of its price and the number of times such edge is used by paths from $r$ in the tree. 
This is the model we will consider.
We next provide the formal model in detail and explain our contribution.

\paragraph{The shortest path tree game.}
We next provide a description of the Stackelberg shortest path tree game.
In fact, we present it as an optimization problem, which we denote by \StackSPT.
The input consists of the following data:
\begin{itemize}
	\item A directed graph $G=(V,E)$ with $n$ vertices and $m$ edges.
	\item A partition of the edges $E$ into $E_F\cup E_P$.
		The edges of $E_P$ are the \DEF{priceable} edges and the edges of $E_F$ are the \DEF{fixed-cost} edges.
	\item A root $r\in V(G)$.
	\item A \DEF{demand} function $\phi:V(G)\rightarrow \RR_{\ge 0}$, where $\phi(v)$ tells the demand 
		of vertex $v$.
	\item A \DEF{cost} function $c:E_F\rightarrow \RR_{> 0}$ fixing the price of the edges in $E_F$.
\end{itemize}

An example is given in Figure~\ref{fig:example}.
A feasible solution is given by a \DEF{price function} $p:E_P\rightarrow \RR_{> 0}$.
The cost function $c$ and the price function $p$ define a weight function $w_p:E\rightarrow \RR_{\ge 0}$
over all edges by setting $w_p(e)=p(e)$ if $e\in E_P$ and $w_p(e)=c(e)$ if $e\in E_F$.
This weight function defines shortest paths in $G$. (In fact, they should be called cheapest paths in this context.)

For a price function $p$ and a path $\pi$, the revenue per unit along $\pi$ is 
\[
	\rho_u(\pi,p) ~:=~ \sum_{e\in E_P\cap E(\pi)} p(e).
\]
Note that only priceable edges contribute to the revenue.
Let $T$ be a subtree of $G$ containing paths from $r$ to all vertices.
For any vertex $v\in V(G)$, let $T[r,v]$ denote the path in $T$ from $r$ to $v$.
The \DEF{revenue} given by $T$ is
\[
	\rho(T,p) ~:=~ \sum_{v\in V(G)} \phi(v)\cdot \rho_u(T[r,v],p).
\]
We would like to tell that the revenue given by the price function $p$
is $\rho(T,p)$, where $T$ is a shortest path tree from $r$ with respect to $w_p$. 
However, there may be different shortest path trees $T$ with
different revenues. In such case, $T$ is taken as the shortest path tree that maximizes the revenue.
Although this assumption may seem counterintuitive at first glance, 
it forces the existence of a maximum and avoids the technicality of attaining revenues
arbitrarily close to a value that is not attainable.
Thus, the revenue of a price function $p$ is defined as
\begin{equation} \label{eq:rho}
	\rho(p) ~:=~ \max\{ \rho(T,p)\mid \mbox{$T$ a shortest path tree in $G$ with respect to $w_p$}\}.
\end{equation}
As an optimization problem, \StackSPT consists of finding a price function $p$ such that 
the revenue $\rho(p)$ is maximized.

From the point of view of game theory, the leader chooses the price function $p$ and the
follower chooses a tree $T$ containing paths from $r$ to all vertices. 
The payoff of the leader is $\rho(T,p)$.
The payoff of the follower is the sum, over all vertices $v$ of $G$, 
of the distance in $T$ from $r$ to $v$.
Among trees $T$ with the same payoff for the follower, she maximizes the revenue $\rho(T,p)$. 
Thus, the follower uses a lexicographic order where, as primary criteria, lengths are minimized, 
and, as secondary criteria, revenue is maximized.

\begin{figure}
\centering
	\includegraphics[scale=1.4,page=1]{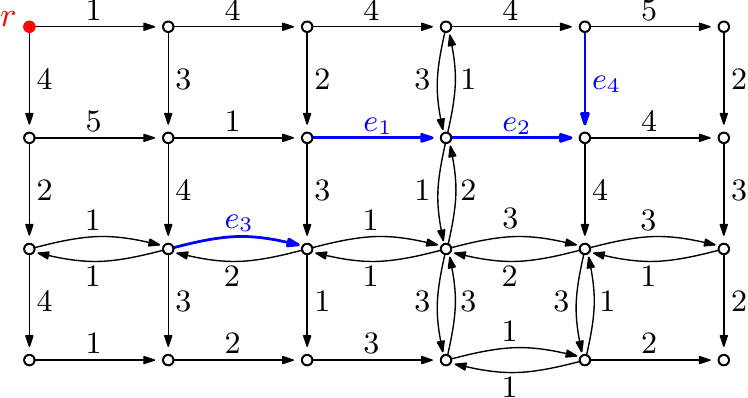}
	\caption{An example of a Stackelberg shortest path tree game. We assume that each vertex has unit demand.}
	\label{fig:example}
\end{figure}

\paragraph{Our result and comparison.}
We assume henceforth that $k:=|E_P|\ge 2$ is a constant. 
For $k=1$, \StackSPT can be solved in $O(m+n\log n)$ time as discussed by Bil\`o et al~\cite{bgpw-08}.

We describe a data structure that can be constructed in $O(m+n\log^{k-1} n)$ time and with the property 
that, given a price function $p$, the revenue $\rho(p)$ can be computed in $O(\log^{k-1} n)$ time.
Bil\`o et al.~\cite{bgpw-08} show how to find an optimal price function $p$
by evaluating the revenue of $O(n^k)$ price functions\footnote{They only discuss the
case when the demand function $\phi$ is identically $1$. 
However, their discussion can be easily adapted to more general demand functions.}.
Combined with our data structure, we can then find an optimal price function in $O(m+n^k\log^{k-1} n)$ time.

Our result matches the result of Bil\`o et al.~\cite{bgpw-08} for the case $k=2$. 
For $k\ge 3$, the algorithm of Bil\`o et al. uses $O(n^k(m + n\log n))$ time. 
A previous algorithm by van Hoesel et al.~\cite{hkmbo} to compute the optimal solution in a more general
Stackelberg pricing problem, where paths from different sources have to be considered, 
reduces \StackSPT to $O(n^{4^k})$ linear programs of constant size. 

The large dependency on $k$ is unavoidable because the problem is NP-hard for unbounded $k$.
Inapproximability results were shown by Joret~\cite{j-11}, and improved by Briest et al.~\cite{betal-10},
for the shortest path between two points. This is a special case of our model where the demand function $\phi$ is 
nonzero for a single vertex.
Briest et al.~\cite{bhk-12} provide an approximation algorithm for more general Stackelberg network pricing games.
When it is specialized to \StackSPT, it provides a $O(\log n)$-approximation.

Our data structure is based on three main ideas:
\begin{itemize} 
\item A careful rule to break ties when there are multiple shortest path trees. 
	With this rule, we can easily split the vertices into groups that use the same priceable edges.
\item Using a smaller network, of size $O(k^2)$,
	such that, for a given price function, we can find out the structure of the priceable edges in
	the shortest path tree of the network. This idea is similar to the
	\emph{shortest paths graph model} of Bouhtou et al.~\cite{bhkl}. 
\item Mapping each vertex of the network to a point in Euclidean $k$-dimensional space 
	in such a way that the vertices that use a certain subset of the priceable edges can be identified 
	as a subset of points in a certain octant. This allows us to use efficient data structures for range searching.
	Similar ideas have been used for graphs of bounded treewidth; 
	see~\cite{isaac2005,ck-2009,poster} and~\cite[Chapter 4]{thesis-shi}.
\end{itemize}

\paragraph{Notation.}
We use $e_1,e_2,\dots,e_k$ to denote the edges of $E_P$, where each edge $e_i=\dart{s_i}{t_i}$.
The enumeration of the edges is fixed; in fact we will use it to break ties.
Perhaps a bit misleading but quite useful, we will use $p(e)=0$ for each $e\in E_F$.
For a subset of vertices $U\subseteq V(G)$ we use the notation $\phi(U):=\sum_{u\in U} \phi(u)$.
For a subset of edges $F\subseteq E$ we use the notation $w_p(F):=\sum_{e\in F} w_p(e)$
and $p(F):=\sum_{e\in F} p(e) = \sum_{e\in F\cap E_F} p(e)$. 

A path $\pi$ will be treated sometimes as a sequence of vertices and sometimes as an edge set.
No confusion can arise from our use.
We use $E_P(\pi)$ for the set of priceable edges along $\pi$, that is, $E_P(\pi)= \pi\cap E_P$. 
Similarly, we use $E_F(\pi)=\pi\cap E_F$ for the fixed-cost edges.
Therefore $w_p(E_P(\pi))=p(\pi\cap E_P)$ and
$w_p(E_F(\pi)) = c(\pi\cap E_F)$.

For any two vertices $u$ and $v$ of $G$
we use $\pi_p(u,v)$ to denote a shortest path from $u$ to $v$ with respect to the weights $w_p$
and $d_p(u,v)$ to denote its weight.
We use $d_0$ as a shorthand for $d_p$ when $p=0$, that is, when the price function assigns price $0$ to each priceable edges.
We use $d_\infty$ as a shorthand for $d_p$ when $p=\infty$, that is, when the price function assigns price $\infty$ to each priceable edge.

For a path $\pi$ and vertices $u,v$ along $\pi$, we use $\pi[u,v]$ for the subpath of $\pi$ from $u$ to $v$.
Similarly, as we have used above, for a tree $T$ and vertices $u,v$, we use $T[u,v]$ to denote the subpath
of $T$ from $u$ to $v$.

\section{Range Searching}
\label{sec:prelim}

Let $X$ be a set of points in $\mathbb{R}^d$. 
Assume we are given a function $\varphi:X\rightarrow \mathbb{R}$
that assigns a weight $\varphi(x)$ to each point $x\in X$. 
We extend the weight function to any subset $Y$ of points by $\varphi(Y):=\sum_{x\in Y} \varphi(x)$. 
A \emph{rectangle} $R$ in $\mathbb{R}^d$ is the Cartesian product 
of $d$ intervals, $R=I_1\times\dots\times I_d$, where each interval $I_i$ can
include both extremes, one of them, or none.

Orthogonal range searching deals with the problem of preprocessing $X$
such that, for a query rectangle $R$, certain properties of $X\cap R$ can be efficiently reported.
We will use the following standard result.

\begin{theorem}[\cite{willard}]\label{theo:range1}
	Let $d\geq 2$ be a constant. 
	Given a set of $n$ points $X\subset\mathbb{R}^d$ and a weight function $\varphi:X\rightarrow \mathbb{R}$, 
	there is a data structure that can be constructed in $O(n\log^{d-1} n)$ time such that,
	for any query rectangle $R$, the weight $\varphi(X\cap R)$ can be reported in $O(\log^{d-1}n)$ time.
\end{theorem}

\section{Breaking Ties}
\label{sec:shortest}

Evaluating the revenue of a price function is easier in a \emph{generic case}, when there is a unique
shortest path from $r$ to each vertex of $V(G)$. In contrast, in the \emph{degenerate case}, there 
is at least one vertex $v$ with two distinct shortest paths from $r$ to $v$.
Unfortunately, the price functions defining the optimum are degenerate. This is easy to see because,
in a generic case, a slight increase in the price function leads to a slight increase in the revenue.

In our approach, we will count how many vertices use a given sequence of priceable edges. 
For this to work, we need a systematic way to break ties, that is, a rule to select, 
among the shortest path trees that give the same revenue, one. 
We actually do not go that far, and only care about the priceable edges on the paths of the tree.
 
We first discuss how to break ties among shortest paths, 
and then discuss how to break ties among shortest path trees.
Essentially, we compare paths lexicographically according to the following:
firstly, we compare paths by length; secondly, if they have the same length,
we compare them by revenue; finally, if they have the same length and revenue,
we compare the priceable edges on the path lexicographically,
giving preference to priceable edges of larger index. 
We next provide the details. 

Define the function $\chi:E\rightarrow \RR_{\ge 0}$ by
$\chi(e_i): = 2^i$, when $e_i\in E_P$, and $\chi(e):=0$ when $e\in E_F$.
We extend the function to subsets of edges by
defining 
\[
	\forall F\subseteq E :~~ \chi(F):=\sum_{e\in F}\chi(e) = \sum_{e_i\in F} 2^i.
\]
Note that, for any two subsets $F$ and $F'$ of priceable edges, $\chi(F)>\chi(F')$
if and only if the edge with largest index in the symmetric difference of $F$ and $F'$ comes from $F$. 
Moreover
\begin{equation}
	\forall F,F'\subseteq E_P: ~~ \chi(F)=\chi(F') \Longleftrightarrow F=F'. \label{eq:chi}
\end{equation}

Define the function 
\begin{align*}
	\w: E & \rightarrow \RR_{>0}\times \RR_{\le 0}\times \RR_{\ge 0} \\
			e	& \mapsto \left( w_p(e), -p(e), -\chi(e) \right)
\end{align*}
Recall that we had set $p(e)=0$ when $e\in E_F$. 
We extend $\w$ to subsets of edges by setting
\[
	\forall F\subseteq E :~~ \w(F) ~:=~ \sum_{e\in F} \w(e).
\]
We treat $\w$ as composite weights that are compared using the lexicographic order $\prec$.
We say that a path $\pi$ is $\w$-shorter than a path $\pi'$ if and only if 
$\w(\pi)\prec \w(\pi')$, where $\prec$ denotes
the lexicographic order. 

For any cycle $\alpha$, the first component of $\w(\alpha)$ is $w_p(\alpha)$, which is positive.
This implies that we do not have ``negative cycles" and we can use 
the weights $\w$ to define $\w$-shortest paths:
a path $\pi$ from $u$ to $v$ is \DEF{$\w$-shortest} if $\w(\pi)$ is minimal,
among the paths from $u$ to $v$, with respect $\prec$.
More compactly:
\[
	\mbox{$\pi$ from $u$ to $v$ is $\w$-shortest}~~\Longleftrightarrow~~ \forall \mbox{ paths $\pi'$ from $u$ to $v$}: ~ \w(\pi) \preceq \w(\pi').
\]

A tree $T$ is a \DEF{$\w$-shortest path tree} (from $r$) if
it contains a $\w$-shortest path from $r$ to each vertex.
Note that this is stronger than telling that $\w(E(T))$ is minimal with respect 
to $\prec$. See Figure~\ref{fig:tree} for an example.
A $\w$-shortest path tree can be computed be computed
in $O(m+ n\log n)$ time using Dijkstra's algorithm with the weights $\w$ and lexicographic comparison\footnote{If
one dislikes using lexicographic comparison, it is also possible to use
weights $w'(e)=w(e)-\varepsilon_1 p(e)- \varepsilon_2\chi(e)$, where $\varepsilon_1=\max_{e} w(e)/n^3$ and
$\varepsilon_2=\varepsilon_1/(2^k n^3)$.}.
(Here we need that $k$ is a constant, which implies that $\chi(F)$ uses $k=O(1)$ bits. 
For general $k$, the running time of Dijkstra's algorithm may get an additional dependence on $k$, depending on the model
of computation.) Note that there may be several $\w$-shortest path trees because of
different shortest paths without priceable edges.

\begin{figure}
\centering
	\includegraphics[scale=1.4,page=3]{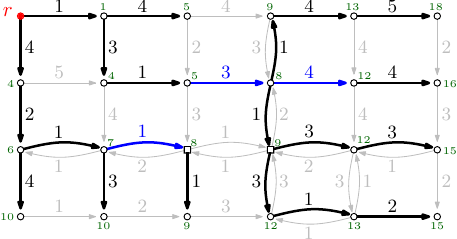}
	\caption{A $\w$-shortest path tree for the price function $p(e_1)=3$, $p(e_2)=p(e_4)=4$, $p(e_3)=1$ 
			in the network of Figure~\ref{fig:example}. The values in the vertices are the distance from $r$.
			Note that there are some vertices, 
			like for example the two that are marked with squares, for which there are different 
			shortest paths using different priceable edges, so we have to select shortest paths maximizing revenue.
			The revenue given by this tree, if each vertex has unit demand, is $p(e_1)\cdot 10 + (p(e_1)+p(e_2))\cdot 2 + p(e_3)\cdot 2=46$ units.}
	\label{fig:tree}
\end{figure}

\begin{lemma}\label{le:rho}
	If $T$ be a $\w$-shortest path tree,
	then $\rho(T,p)=\rho(p)$.
\end{lemma}
\begin{proof}
	Since $T$ is a $\w$-shortest path tree, it is also a shortest path tree for the weights $w_p$.
	By the definition of $\rho(p)$ given in equation~(\ref{eq:rho}), we have $\rho(T,p)\le \rho(p)$.
	We next show that $\rho(T,p)\ge \rho(p)$, which implies that $\rho(T,p) = \rho(p)$.
	
	Consider a shortest path tree $T^*$ that defines the value $\rho(p)$.
	That is, $\rho(T^*,p)=\rho(p)$.
	Since $T$ is a $\w$-shortest path tree, we have
	\begin{equation}\label{eq:T1}
		\forall v\in V(G):~~  \w (T[r,v])\preceq \w(T^*[r,v]).
	\end{equation}
	Since $T$ and $T^*$ are both shortest path trees, we have 
	\begin{equation}\label{eq:T2}
		\forall v\in V(G):~~   w_p(T[r,v])= w_p(T^*[r,v]).
	\end{equation}
	Expanding the definition of $\w$, from (\ref{eq:T1}) and (\ref{eq:T2}) 
	we obtain
	\[
		\forall v\in V(G):~~   p(T[r,v]) \ge p(T^*[r,v]).
	\]
	This means that 
	\[
		\rho(p) ~=~ \rho(T^*,p) ~=~ \sum_{v\in V(G)} \phi(v)\cdot p(T^*[r,v])
					 ~\le~ \sum_{v\in V(G)} \phi(v)\cdot p(T[r,v])
					 ~=~ \rho(T,p).
	\]
\end{proof}

\section{Reduced trees and sequences of priceable edges}
\label{sec:reduced}

Consider a price function $p$. Let $T$ be a $\w$-shortest path tree from $r$.
The \DEF{$\w$-reduced tree} $RT$ is obtained from $T$ by contracting
all the fixed-cost edges $E_F\cap E(T)$. The resulting graph is a tree with
edge set $E_P\cap E(T)$. When considering $RT$, we disregard the prices $p$ and the orientation
of the edges, and consider it as a rooted, unweighted, undirected graph with distinct labels $e_1,\dots,e_k$ on its edges.
In general, we will use $RH$ to denote the reduced graph obtained from a graph $H$ by contracting all non-priceable edges.
The $\w$-reduced tree for the example of Figure~\ref{fig:tree} contains the edges $e_1$ and $e_3$ adjacent to $r$
and the edge $e_2$ below $e_1$.

We first show that the $\w$-reduced trees are independent of the $\w$-shortest path
tree that is used.

\begin{lemma}\label{le:reduced}
	If $T$ and $T'$ are $\w$-shortest path trees, then $RT=RT'$.
\end{lemma}
\begin{proof}
	Since both $T$ and $T'$ are $\w$-shortest path trees we have
	\[
		\forall v\in V(G):~~ \w (T[r,v])= \w(T'[r,v]),
	\]
	which means
	\begin{align*}
		\forall v\in V(G):~~  w_p(T[r,v])&= w_p(T'[r,v]), \\
							  p(T[r,v])&= p(T'[r,v])), \\
							  \chi (T[r,v])&= \chi (T'[r,v]).
	\end{align*}
	From the last equality and the property (\ref{eq:chi}) we have 
	\[
		\forall v\in V(G):~~ E_P(T[r,v])= E_P(T'[r,v]).
	\]
	If $e_i=\dart{s_i}{t_i}$ is a descendant of $e_j$ in $T$, this means
	that $e_j\in E_P(T[r,s_i])$ and $e_j\in E_P(T[r,t_i])$. But then for $T'$ 
	we also have $e_j\in E_P(T'[r,s_i])$ and $e_j\in E_P(T'[r,t_i])$,
	which implies that $e_i$ is a descendant of $e_j$ in $T'$. 
	By symmetry, we conclude that $e_i$ is a descendant of $e_j$ in $T$ if
	and only if $e_i$ is a descendant of $e_j$ in $T'$. This implies
	that the $\w$-reduced trees $RT$ and $RT'$ are the same. 
\end{proof}

A useful consequence of this is that any two $\w$-shortest path trees have the 
same subset of priceable edges.

\begin{lemma}\label{structure1}
	In $O(m+n\log n)$ time we can construct a data structure with the property
	that, for any given price function $p$, we can compute in $O(1)$ time the $\w$-reduced tree $RT$.
\end{lemma}
\begin{proof}
	We construct the \emph{model graph} $\tilde G=\tilde G(G,E_P,c,r)$, as follows. 	
	The vertex set of $\tilde G$ consists of $r$ and the endpoints of the priceable edges. 
	Thus $V(\tilde G)=\{r \}\cup \{ s_1,t_1,\dots,s_k,t_k\}$.
	In $\tilde G$, we have edges from $r$ to any other vertex.
	Furthermore, for each priceable edges $e_i$ and $e_j$, $i\not= j$,
	we have an edge from $t_i$ to $s_j$ and to $t_j$. 
	Finally, we have the edges $e_1,\dots, e_k$ themselves.
	 
	\begin{figure}
	\centering
		\includegraphics[width=\textwidth,page=2]{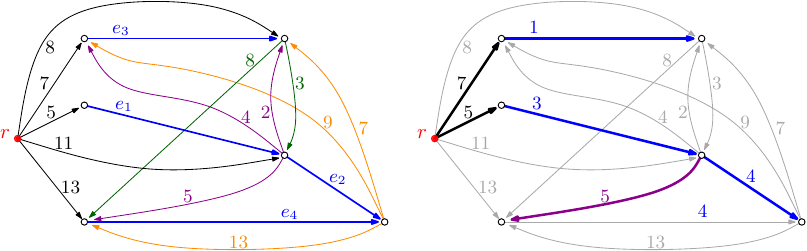}
		\caption{Left: The model graph for the network of Figure~\ref{fig:example}. Edges with infinite weight, like
			for example $\dart{r}{t_4}$ or $\dart{t_1}{t_2}$, are not drawn. 
			Right: the $\w$-shortest path tree in the model for the price function of Figure~\ref{fig:tree}: 
			$p(e_1)=3$, $p(e_2)=p(e_4)=4$, and $p(e_3)=1$.}
		\label{fig:model}
	\end{figure}

	Each edge $\dart{u}{v}$ in $E(\tilde G)$ that is not a priceable edge
	gets weight $d_\infty(u,v)$. That is, each edge $\dart{u}{v}$
	gets weight equal to the distance between $u$ and $v$ in $G-E_P$.
	This finishes the description of the model graph $\tilde G$.
	See Figure~\ref{fig:model}, left, for an example.
	This construction is similar to and inspired by the \emph{shortest paths graph model} 
	of Bouhtou et al.~\cite{bhkl}.
	
	The model graph $\tilde G$ has the same priceable edges as $G$.
	Consider any price function $p$. We \emph{claim that the $\w$-reduced trees 
	for $G$ and $\tilde G$ are the same}. That is, if $\tilde T$ denotes a
	$\w$-shortest path tree in $\tilde G$ and $R\tilde T$ denotes the $\w$-reduced tree obtained
	after contracting all non-priceable edges, then $R\tilde T= RT$.
	See Figure~\ref{fig:model}, right, for an example.

	Consider the subgraph $\tilde F$ of $G$ obtained by ``expanding" each shortest path
	of $\tilde T$: for each priceable edge $e_i$ in $\tilde T$  we put the same edge in $\tilde F$;
	for each non-priceable edge $\dart{u}{v}$ of $\tilde T$ we put
	in $\tilde F$ a shortest path of $G-E_P$ that connects $u$ to $v$.
	The graph $\tilde F$ is like a $\w$-shortest path forest spanning the vertices of $V(\tilde G)$.
	Any path in $\tilde T$ corresponds to a path in $\tilde F$ with the same composite weight $\w$.
	The reduced tree $R\tilde T$ is obtained from $\tilde F$ by contracting the fixed-cost edges $E(\tilde F)\cap E_F$.
	That is, $R\tilde F= R\tilde T$.
	
	Consider any edge $e_i\in E_P$.
	Since $T[r,t_i]$ is a $\w$-shortest path in $G$, we have 
	\[
		\w (T[r,t_i]) ~\prec~ \w (\tilde F[r,t_i]).
	\]
	Since $\tilde T [r,t_i]$ is a $\w$-shortest path in $\tilde G$, we have
	\[
		\w (\tilde T[r,t_i]) ~=~  \w (\tilde F[r,t_i]) ~\prec~ \w (T[r,t_i]).
	\]
	We thus conclude that
	\[
		\w (T[r,t_i]) ~=~ \w (\tilde F[r,t_i]).
	\]
	This means that 
	\[
		\chi(T[r,t_i]) ~=~ \chi(\tilde F[r,t_i]),
	\]
	and by (\ref{eq:chi}) we get that 
	\[
		E_P (T[r,t_i]) ~=~ E_P(\tilde F[r,t_i]) ~=~ E_P(\tilde T[r,t_i]).
	\]
	The same discussion for $s_i$ implies that
	\[
		E_P (T[r,s_i]) ~=~ E_P(\tilde F[r,s_i]) ~=~ E_P(\tilde T[r,s_i]).
	\]
	Since the same holds for each priceable edge $e_i$, it follows that $RT$ and $R\tilde T$.
	Indeed, if $e_j$ is a descendant of $e_i$, then $e_j\in E_P (T[r,s_i])$ and $e_j\in E_P (T[r,t_i])$,
	which means that $e_j\in E_P (\tilde T[r,s_i])$ and $e_j\in E_P (\tilde T[r,t_i])$,
	and we conclude that $e_i$ is a descendant of $e_j$ in $\tilde T$.
	This finishes the proof of the claim.
	
	Since $RT=R\tilde T$ and $R\tilde T$ can be computed in constant
	time because $\tilde G$ has constant size, the result follows. 
\end{proof}

Let $\T$ be the family of all possible reduced trees, over all possible graphs $G$.
Thus $\T$ is the family of rooted trees with at most $k$ edges where each edge has a distinct
label among $e_1,\dots,e_k$. It is clear that the number of such trees depends only on $k$,
and thus it is bounded by a constant in our case.  

Consider any reduced tree $R\in \T$. Each edge $e_i$ that appears in $R$ defines a \DEF{sequence of priceable edges}, 
denoted by $\sigma(e_i,R)$, which is the sequence of edges followed by the path in $R$ from the root to $e_i$.
The edge $e_i$ is the last edge of $\sigma(e_i,R)$. When $e_i$ is not in $R$ we define
$\sigma(e_i,R)$ as the empty sequence. Since each edge of $R$
defines a different sequence of edges, the tree $R\in \T$ defines $|E(R)|$ nonempty sequences.

For any nonempty sequence $\sigma=(e_{i_1},\dots, e_{i_a})$ of distinct priceable edges, 
we define 
\begin{align*}
	W_\infty(\sigma) & ~:=~ d_\infty(r,s_{i_1})+\sum_{1\le j\le a-1} d_\infty (t_{i_j},s_{i_{j+1}}),\\
	p(\sigma) &~:=~ \sum_{e\in \sigma} p(e) ~=~ \sum_{1\le j\le a} p(e_{i_j}),\\
	\chi(\sigma) &~:=~ \sum_{e\in \sigma} \chi(e) ~=~ \sum_{1\le j\le a} \chi(e_{i_j}) .
\end{align*}
We define an order $\prec_p$ among sequences of priceable edges in a reduced tree $R\in \T$.
For sequences $\sigma,\sigma'$ in $R$, it holds $\sigma\prec_p \sigma'$ if 
and only if $(p(\sigma),\chi(\sigma))\succ (p(\sigma'),\chi(\sigma'))$,
where $\succ$ denotes the lexicographic comparison.
Therefore
\[
	\sigma\prec_p \sigma' ~~\Longleftrightarrow~~ p(\sigma)>p(\sigma') \mbox{ or } 
							\left( p(\sigma)=p(\sigma') \mbox{ and } \chi(\sigma)> \chi(\sigma')	\right)
\]
Because of property (\ref{eq:chi}), $\prec_p$ is a linear order among the sequences of priceable edges
in a reduced tree. That is, for any two distinct sequences $\sigma$ and $\sigma'$, 
either $\sigma\prec_p \sigma'$ or $\sigma' \prec_p \sigma$.

Consider any two paths $\pi$ and $\pi'$ in $G$ with sequences of priceable edges $\sigma$ and $\sigma'$, respectively.
Because of the definition of $\w$ and $\prec_p$ we have
\begin{equation}\label{eq:5}
	\w (\pi)\prec \w(\pi') ~~\Longleftrightarrow~~ w_p(\pi) < w_p(\pi') \mbox{ or } 
							\left( w_p(\pi) = w_p(\pi') \mbox{ and } \sigma\prec_p \sigma'	\right)
\end{equation}

Consider any shortest path $\pi$ from $r$ to $t_i$. If the priceable edges that
appear along $\pi$ follow the sequence $\sigma=(e_{i_1},\dots, e_{i_a})$, where $e_{i_a}=e_i$,
we can then decompose $\pi$ into the subpaths 
\[
	\pi[r,s_{i_1}] ,~ e_{i_1},~ \pi[t_{i_1},s_{i_2}],e_{i_2}, \dots, \pi[t_{i_{a-1}},s_{i_a}],~ e_{i_{a}}.
\]
Since each of those subpaths is shortest, the length of $\pi$ is 
\[
	d_\infty(r,s_{i_1}) + p(e_{i_1}) + d_\infty(t_{i_1},s_{i_2}) + p(e_{i_2}) + \dots +
		 d_\infty(t_{i_{a-1}},s_{i_a}) + p(e_{i_{a}}),
\]	
and thus
\begin{equation}\label{eq:1}
		w_p(\pi) ~=~ W_\infty(\sigma) + p(\sigma).
\end{equation}

\section{Data structure for computing the revenue}
\label{sec:datastructure}

Consider a price function $p$ and let $T$ be a $\w$-shortest path tree.
For each edge $e_i\in E_P$, let
$V_{T}(e_i,p)$ be the set of vertices with the property that $e_i$ is the last edge of $E_P$
used by $T[r,v]$. It may be that $V_{T}(e_i,p)=\emptyset$. In particular this happens when $e_i$
does not appear in the shortest path tree $T$. We first argue that $V_{T}(e_i,p)$ is 
independent of the choice of $T$.

\begin{lemma}\label{le:equal}
	If $T$ and $T'$ are $\w$-shortest path trees, then,
	for each $e_i\in E_P$, it holds that $V_{T}(e_i,p)=V_{T'}(e_i,p)$.
\end{lemma}
\begin{proof}
	Since both $T$ and $T'$ are $\w$-shortest path trees, we have seen in the
	proof of Lemma~\ref{le:reduced} that
	\[
		\forall v\in V(G):~~ E_P(T[r,v])= E_P(T'[r,v]).
	\]
	
	Consider any vertex $v\in V_{T}(e_i,p)$. Since $e_i\in E_P(T[r,v])= E_P(T'[r,v])$,
	there is some edge $e_j$ such that $v\in V_{T'}(e_j,p)$. 
	We want to show that $e_j=e_i$. This will imply that $V_{T}(e_i,p)\subseteq V_{T'}(e_i,p)$,
	and by symmetry we have equality, as stated.
	
	Since $v\in V_{T}(e_i,p)$ we have $E_P(T[r,v])=E_P(T[r,t_i])$.
	Similarly, we have $E_P(T'[r,v])=E_P(T'[r,t_j])$ because $v\in V_{T'}(e_j,p)$.
	Putting it together we have 
	\[
		E_P(T[r,t_i]) ~= E_P(T[r,v]) ~=~ E_P(T'[r,v]) ~=~ E_P(T'[r,t_j]) ~=~ E_P(T[r,t_j]).
	\]
	By Lemma~\ref{le:reduced}, $E(T)\cap E_P= E(T')\cap E_P$, which means
	that $e_j$ is also an edge of $T$. The equality $E_P(T[r,t_i])=E_P(T[r,t_j])$
	then implies that $e_i=e_j$.
\end{proof}

Since $V_{T}(e_i,p)$ is independent of the $\w$-shortest path tree $T$ that is used,
we will just denote it by $V(e_i,p)$.

\begin{lemma}\label{le:formula}
	Let $p$ be a price function and $R$ its $\w$-reduced tree. The revenue given by $p$ is
	\[
		\rho(p) ~=~ \sum_{e_i\in E(R)} p(\sigma(e_i,R)) \cdot \phi(V(e_i,p)).
	\]
\end{lemma}
\begin{proof}
	Note that, if $i\not= j$, then $V(e_i,p)$ and $V(e_j,p)$ are disjoint by definition.
	Let us set $V_0=V\setminus (\cup_{e_i\in E(R)} V(e_i,p))$.
	The vertices of $V_0$ do not contribute anything to the revenue because the corresponding
	paths do not use any priceable edges.
	
	Let $T$ be a $\w$-shortest path tree and let $R$ be the $\w$-reduced tree. We then have $p(T[r,v])=0$ for all $v\in V_0$.
	Using the definition of $\rho(T,v)$ and that $V_0,V(e_1,p),\dots,V(e_k,p)$ is partition of $V(G)$
	we have
	\begin{align*}
		\rho(T,p) & ~=~ \sum_{v\in V(G)} \phi(v) \cdot p(T[r,v]) \\
				  & ~=~ \sum_{v\in V(G)\setminus V_0} \phi(v) \cdot p(T[r,v]) \\
				  & ~=~ \sum_{e_i\in E(R)} ~ \sum_{v \in V(e_i,p)} \phi(v) \cdot p(T[r,v]) \\
				  & ~=~ \sum_{e_i\in E(R)} ~ \sum_{v \in V(e_i,p)} \phi(v) \cdot p(\sigma(e_i,R))\\
				  & ~=~ \sum_{e_i\in E(R)} p(\sigma(e_i,R)) \cdot \phi(V(e_i,p)).
	\end{align*}
	where in the fourth equality we have used that $p(T[r,v])= p(\sigma(e_i,R))$ for all vertices of $V(e_i,p)$.
	Since $\rho(p)=\rho(T,p)$ because of Lemma~\ref{le:rho}, we conclude that 
	\begin{equation}
		\rho(p) ~=~ \sum_{e_i\in E(R)} p(\sigma(e_i,R)) \cdot \phi(V(e_i,p)).
	\end{equation}
\end{proof}

Our objective is to compute $\phi(V(e_i,p))$ efficiently using data structures.
Since all vertices in $V(e_i,p)$ use the same priceable edges, 
this will lead to an efficient computation of the revenue.

\begin{lemma}\label{structure2}
	Assume that $k\ge 2$ is a constant.
	Consider a reduced tree $R\in \T$ and an edge $e_i\in E(R)$.
	In time $O(m+n\log^{k-1} n)$ we can construct a data structure with the following property: 
	given a price function $p$ with the property that its $w_p$-reduced tree is $R$,
	we can obtain $\phi(V(e_i,p))$ in $O(\log^{k-1} n)$ time.
\end{lemma}
\begin{proof}
	For each vertex $v\in V(G)$ we define a point $p_v\in \RR^{k}$ whose $j$th coordinate is
	\[
		p_v(j):=\begin{cases}
				W_\infty(\sigma(e_i,R)) + d_\infty(t_i,v) - W_\infty(\sigma(e_j,R)) - d_\infty(t_j,v)\quad & 
																			\mbox{if $j\not= i$;}\\
				W_\infty(\sigma(e_i,R)) + d_\infty(t_i,v) - d_\infty(r,v) & \mbox{if $j=i$.}
				\end{cases}
	\]
	Let $P$ be the set of points $\{ p_v \mid v \in V(G)\}$.
	To each point $p_v\in P$ we assign the weight $\varphi(p_v):=\phi(v)$. 
	We then store the point set $P$ using the data structure for range searching of Theorem~\ref{theo:range1}.
	This finishes the description of the data structure. 
	
	The construction of the data structure takes $O(m+n\log^{k-1} n)$ time. 
	We first run a shortest path tree algorithm in $G-E_P$ from $r$ and from each endpoint of the edges in $E_P$. 
	This takes $O(k(m+n\log n))=O(m+n\log n)$ time.
	With this information we can obtain each coordinate of each point in constant time,
	and thus we construct $P$ in $O(n)$ time. The construction of the data structure of Theorem~\ref{theo:range1}
	takes $O(n\log^{k-1} n)$ time because we have $k$-dimensional points.
	
	Consider a price function $p$ and let $T$ be a $\w$-shortest path tree.
	By assumption, $RT=R$. 
	We next explain how to recover $\phi(V(e_i,p))$.
	For every $j=1,\dots ,k$, define the interval 
	\[
		I(j)= \begin{cases}
					\left( -\infty, p(\sigma(e_j,R))- p(\sigma(e_i,R))\right] & \mbox{if $i\not=j$ and $\sigma(e_i,R)\prec_p \sigma(e_j,R)$;}\\
					\left( -\infty, p(\sigma(e_j,R))- p(\sigma(e_i,R))\right) & \mbox{if $i\not=j$ and $\sigma(e_i,R)\not\prec_p \sigma(e_j,R)$;}\\
					\left( -\infty, -p(\sigma(e_i,R)) \right]									  & \mbox{if $i=j$.}
				\end{cases}
	\]
	Consider a vertex $v\in V(G)$. The path $T[r,v]$ can be disjoint from $E_P$ or follow 
	one of the sequences $\sigma(e_1,R),\dots,\sigma(e_k,R)$.
	Using the relation in equation (\ref{eq:5}), we see that the path $T[r,v]$ 
	follows the sequence $\sigma(e_i,R)$ if and only if the following conditions hold:
	\begin{align*}
		& w_p(T[r,t_i])+ d_\infty (t_i,v) ~\le~ d_\infty (r,v),\\
		&\forall j\not= i, \sigma(e_i,R)\prec_p \sigma(e_j,R) : \\
			& ~~~~~~~~~~~~~~~~~~~~~~~~~~~~~ w_p(T[r,t_i])+ d_\infty (t_i,v) ~\le~ w_p(T[r,t_j])+ d_\infty (t_j,v);
		\\
		&\forall j\not= i, \sigma(e_i,R)\not\prec_p \sigma(e_j,R) : \\
			& ~~~~~~~~~~~~~~~~~~~~~~~~~~~~~ w_p(T[r,t_i])+ d_\infty (t_i,v) ~<~ w_p(T[r,t_j])+ d_\infty (t_j,v).
	\end{align*}
	Because of equation (\ref{eq:1}), we have, for each $j=1,\dots,k$,
	\[
		w_p(T[r,t_j]) ~=~ p(\sigma(e_j,R)) + W_\infty(\sigma(e_j,R)).
	\]
	and thus the conditions are equivalent to: 
	\begin{align*}
		& p(\sigma(e_i,R)) + W_\infty(\sigma(e_i,R)) + d_\infty (t_i,v) ~\le~ d_\infty (r,v),\\
		&\forall j\not= i, \sigma(e_i,R)\prec_p \sigma(e_j,R) : \\
			& ~~~~~~~~~~ p(\sigma(e_i,R)) + W_\infty(\sigma(e_i,R))+ d_\infty (t_i,v) ~\le ~
					p(\sigma(e_j,R)) + W_\infty(\sigma(e_j,R))+ d_\infty (t_j,v);
		\\
		&\forall j\not= i, \sigma(e_i,R)\not\prec_p \sigma(e_j,R) : \\
			& ~~~~~~~~~~ p(\sigma(e_i,R)) + W_\infty(\sigma(e_i,R))+ d_\infty (t_i,v) ~< ~
					p(\sigma(e_j,R)) + W_\infty(\sigma(e_j,R))+ d_\infty (t_j,v).
	\end{align*}	
	Reordering, we obtain that $v\in V(e_i,p)$ if and only if
	\begin{align*}
		& W_\infty(\sigma(e_i,R))+ d_\infty (t_i,v) - d_\infty (r,v) ~\le~ -p(\sigma(e_i,R)),\\
		&\forall j\not= i, \sigma(e_i,R)\prec_p \sigma(e_j,R) : \\
			& ~~~~~~~~~~ W_\infty(\sigma(e_i,R))+ d_\infty (t_i,v) - W_\infty(\sigma(e_j,R)) - d_\infty (t_j,v) ~\le~
					p(\sigma(e_j,R))- p(\sigma(e_i,R));
		\\
		&\forall j\not= i, \sigma(e_i,R)\not\prec_p \sigma(e_j,R) : \\
			& ~~~~~~~~~~ W_\infty(\sigma(e_i,R))+ d_\infty (t_i,v) - W_\infty(\sigma(e_j,R)) - d_\infty (t_j,v) ~< ~
		p(\sigma(e_j,R))- p(\sigma(e_i,R)).
	\end{align*}
	This condition is equivalent to 
	\begin{align*}
		\mbox{for $j=1,\dots,k$} : ~ p_v(j) \in I(j).
	\end{align*}
	We conclude that $v\in V(e_i,p)$ if and only if $p_v\in \prod_{j} I(j)$.
	We can then recover $\phi(V(e_i,p))=\varphi\left( P\cap \prod_{j} I(j) \right)$ by querying the data structure 
	for $\varphi \left( P\cap \prod_{j} I(j) \right)$.
	
	Given a price function $p$, we can compute the values $p(\sigma(e_j,R))$, for $j=1,\dots, k$, in
	$O(1)$ time. With this information we can compute the extremes of the intervals $I(j)$ and
	query the data structure for range searching in $O(\log^{k-1} n)$ time.
\end{proof}

\begin{theorem}
	Assume that $k\ge 2$ is a constant.
	Consider an instance to $\StackSPT$ with $n$ vertices, $m$ edges, and $k$ priceable edges.
	In time $O(m+n\log^{k-1} n)$ we can construct a data structure with the following property: 
	given a price function $p$, the revenue $\rho (p)$ can be obtained in $O(\log^{k-1} n)$ time.
\end{theorem}
\begin{proof}
	We start constructing the data structure of Lemma~\ref{structure1}, so that we can
	quickly compute the $\w$-reduced tree for any given price function.
	For each reduced tree $R\in \T$ and each priceable edge $e_i\in E(R)$ 
	we construct the data structure from Lemma~\ref{structure2} and denote it by $DS(R,e_i)$.
	This finishes the construction of the data structure.
	The time bound follows from the time bounds of Lemmas~\ref{structure1} and~\ref{structure2}.
	
	Consider a price function $p$. Because of Lemma~\ref{le:formula}, we have
	\begin{equation*}
		\rho(p) ~=~ \sum_{e_i\in E(R)} p(\sigma(e_i,R)) \cdot \phi(V(e_i,p)).
	\end{equation*}
	The data to apply this formula can be recovered from the data structures.
	Firstly, we use the data structure of Lemma~\ref{structure1} to compute the $\w$-reduced tree $R$ for $p$.
	For each $e_i\in E(R)$, we query $DS(R,e_i)$ to recover $\phi(V(e_i,p))$.
	Finally, we compute $p(\sigma(e_i,R))$ for each $e_i\in E(R)$.
	Overall, we use $O(1)$ queries to the data structures and each such query takes $O(\log^{k-1} n)$ time.
	The result follows.
\end{proof}

\begin{corollary}
	Let $k\ge 2$ be a constant.
	The problem \StackSPT with $n$ vertices, $m$ edges, and $k$ priceable edges can be solved in $O(m+n^k \log^{k-1} n)$ time.
\end{corollary}
\begin{proof}
	As discussed in the introduction, Bil{\`o} et al.~\cite{bgpw-08} show how to solve \StackSPT
	by finding the revenue of $O(n^k)$ price functions. Using the theorem we, can
	find the revenue for all those price functions in $O(n^k \log^{k-1} n)$ time after $O(m+n \log^{k-1} n)$ preprocessing time.
\end{proof}

\bibliographystyle{plain}
\bibliography{bibliography}

\begin{thebibliography}{10}

\bibitem{isaac2005}
B.~{Ben-Moshe}, B.~K. Bhattacharya, and Q.~Shi.
\newblock Efficient algorithms for the weighted 2-center problem in a cactus
  graph.
\newblock In X.~Deng and D.~Du, editors, {\em Proc. ISAAC 2005}, volume 3827 of
  {\em LNCS}, pages 693--703. Springer, 2005.

\bibitem{bgpw-08}
D.~Bil{\`o}, L.~Gual{\`a}, G.~Proietti, and P.~Widmayer.
\newblock Computational aspects of a 2-player {S}tackelberg shortest paths tree
  game.
\newblock In {\em Proc. WINE 2008}, volume 5385 of {\em LNCS}, pages 251--262.
  Springer, 2008.

\bibitem{bhkl}
M.~Bouhtou, S.~P.~M. {van Hoesel}, A.~F. {van der Kraaij}, and J.-L. Lutton.
\newblock Tariff optimization in networks.
\newblock {\em INFORMS Journal on Computing}, 19(3):458--469, 2007.

\bibitem{betal-10}
P.~Briest, P.~Chalermsook, S.~Khanna, B.~Laekhanukit, and D.~Nanongkai.
\newblock Improved hardness of approximation for {S}tackelberg shortest-path
  pricing.
\newblock In {\em WINE 2010}, volume 6484 of {\em LNCS}, pages 444--454.
  Springer, 2010.

\bibitem{bhk-12}
P.~Briest, M.~Hoefer, and P.~Krysta.
\newblock {S}tackelberg network pricing games.
\newblock {\em Algorithmica}, 62(3-4):733--753, 2012.

\bibitem{ck-2009}
S.~Cabello and C.~Knauer.
\newblock Algorithms for graphs of bounded treewidth via orthogonal range
  searching.
\newblock {\em Computational Geometry: Theory and Applications},
  42(9):815--824, 2009.

\bibitem{cetal-11}
J.~Cardinal, E.~D. Demaine, S.~Fiorini, G.~Joret, S.~Langerman, I.~Newman, and
  O.~Weimann.
\newblock The {S}tackelberg minimum spanning tree game.
\newblock {\em Algorithmica}, 59(2):129--144, 2011.

\bibitem{j-11}
G.~Joret.
\newblock {S}tackelberg network pricing is hard to approximate.
\newblock {\em Networks}, 57(2):117--120, 2011.

\bibitem{lms-98}
M.~Labb{\'e}, P.~Marcotte, and G.~Savard.
\newblock A bilevel model of taxation and its application to optimal highway
  pricing.
\newblock {\em Management Science}, 44:1608--1622, 1998.

\bibitem{or-94}
M.~J. Osborne and A.~Rubinstein.
\newblock {\em A Course in Game Theory}.
\newblock MIT Press, 1994.

\bibitem{poster}
Q.~Shi.
\newblock Single facility location problems in partial $k$-trees, 2005.
\newblock Poster at MITACS, Canada.

\bibitem{thesis-shi}
Q.~Shi.
\newblock {\em Efficient Algorithms for Network Center/Covering Location
  Optimization Problems}.
\newblock PhD thesis, School of Computing Science, Simon Fraser University,
  2008.

\bibitem{h-08}
S.~P.~M. {van Hoesel}.
\newblock An overview of {S}tackelberg pricing in networks.
\newblock {\em European Journal of Operational Research}, 189(3):1393--1402,
  2008.

\bibitem{hkmbo}
S.~P.~M. {van Hoesel}, A.~F. {van der Kraaij}, C.~Mannino, M.~Bouhtou, and
  G.~Oriolo.
\newblock Polynomial cases of the tarification problem.
\newblock Research Memoranda RM03063, Maastricht University, METEOR, 2003.

\bibitem{willard}
D.~E. Willard.
\newblock New data structures for orthogonal range queries.
\newblock {\em SIAM J. Comput.}, 14:232--253, 1985.

\end{thebibliography}
\end{document}